\documentclass[conference]{IEEEtran}

\usepackage{amsthm,amssymb,graphicx,multirow,amsmath,color,amsfonts}
\usepackage[update,prepend]{epstopdf}
\usepackage[noadjust]{cite}
\usepackage[latin1]{inputenc}
\usepackage{tikz}
\usepackage{bbm} 
\usepackage{pdfpages}
\usepackage{tabulary}
\usepackage{multirow}
\usepackage{comment}
\usepackage{textcomp}

\begin{document}
\title{
Cross-layer Interference Modeling for 5G MmWave Networks in the Presence of Blockage}
\author{\IEEEauthorblockN{Solmaz Niknam, Reza Barazideh, and Balasubramaniam Natarajan\\}
\IEEEauthorblockA{Department of Electrical and Computer Engineering\\
Kansas State University, Manhattan, KS, 66506 USA\\
Email: \{slmzniknam, rezabarazideh, bala\}@ksu.edu\\}
}

\maketitle

\begin{abstract}
Fifth generation (5G) wireless technology is expected to utilize highly directive antennas at millimeter wave (mmWave) spectrum to offer higher data rates. However, given the high directivity of antennas and adverse propagation characteristics at mmWave frequencies, these signals are very susceptible to obstacles. One of the important factors that are highly impacted is interference behavior. In fact, signals received from other terminals can be easily blocked or attenuated at the receiver. In addition, higher number of terminals can transmit signals without introducing much interference and hence the traffic behavior, maintained by medium access control (MAC) layer,  may change.
In this paper, we provide an interference model to evaluate the interference power received at the physical layer of the receiving terminal, considering antenna directivity, effect of obstacles and MAC layer constraints that control the number of terminals transmitting simultaneously. We first develop a blockage model and then derive the Laplace transform of the interference power received at a typical receiving node. Subsequently, using the derived Laplace transform, we evaluate the network error performance using average bit-error-rate (BER). Analytical results are validated via Monte-Carlo simulations.
\end{abstract}

\section{Introduction} \label{sec:intro}
Utilization of millimeter wave (mmWave) spectrum in the range of 30--300~GHz is regarded as a prospective option for future fifth generation (5G) systems. However, mmWave signals suffer from severe pathloss and strong atmospheric absorption. Therefore, highly directional transmissions are necessary for communication in these frequencies~\cite{Andrew2014what}. Using directional antennas has posed multiple challenges in different aspects of 5G communication. In fact, mmWave signals are highly sensitive to blockages. The sensitivity to obstacles in turn impacts the interference behavior~\cite{Niknam2017Spatial-Spectral,Solmaz2017Finite}.
In addition, it is clear that the interference phenomenon happens at the physical layer of the receiving node. However, the interference signal and its undesired effects are impacted by features of the interfering nodes at different network layers. 
Network operation and traffic behavior that describe the transmitter activity and the interrelation among terminals are maintained by the medium access control (MAC) protocols. Therefore, an efficient and comprehensive interference model for mmWave applications must capture both mmWave physical layer specifications and MAC layers constraints. Such cross-layer models can guide the design and development of interference coordination and management schemes~\cite{Agival2016survey5G,Amiri2018self}.

There have been couple of candidate MAC layer protocols for 5G mmWave applications. Among them, multisuer MAC protocols that are based on directional carrier sense multiple access with collision avoidance (D-CSMA/CA) have attracted attention~\cite{Agival2016survey5G}. Such protocols effectively increase the network capacity by exploiting the spatial features.
%
In networks with sensing mechanism, the spatial distribution of active (simultaneously transmitting) transmitters is typically modeled as Matern point process (MPP)~\cite{Haenggi2012stochBook}. MPP can be viewed as a thinned version of Poisson point process (PPP) where it considers an exclusion area (circular area with radius proportional to the sensing threshold of the transmitter) around each node, and all nodes closer than a certain distance are excluded.
In fact, in MPP networks, active nodes are separated by a specific range from each other. However, in real scenarios in mmWave networks with directional signals (that can be easily blocked by obstacles), two transmitters can be close to each other and transmit in two different directions without introducing interference to each other. In addition, they may not be interfering with each other due to the presence of obstacles in between. Therefore, considering antenna directionality and blockage effect, not all the nodes in the circular exclusion area should be eliminated from the point process. 
There have been several prior works on modeling CSMA/CA networks using MPP~\cite{MAC2007Baccelli,CSMA2014Alfano} or modified versions of MPP~\cite{busson2014capacity,lei2016modified,tong2015stochastic,CSMA2013ElSawy}. However, non of them consider the directionality of the antennas which makes them unsuitable for mmWave applications. Although majority of works on the performance evaluation of D-CSMA/CA networks are via simulations, authors in~\cite{DCSMA2010Bazan,Dir2003Wang,DIRECT2006Hsu,Adhoc2006Carvalho} have conducted analytical evaluation of networks with directional MAC protocol. However, blockage effect resulting from high directionality of mmWave signals are not taken into account. Moreover, spatial distribution of access points (APs) and random orientation of the antennas are not considered in~\cite{DIRECT2006Hsu,Adhoc2006Carvalho}, as well.

In this paper, we propose a cross-layer interference model that considers both directionality of mmWave signals with random antenna orientation and blockage effect from both physical and MAC layer perspective. Using tools from stochastic geometry, we model the spatial distribution of APs and blockages as Poisson processes. Transmitting nodes are equipped with very directional antennas towards their intended receivers. We also assume that APs employ sensing mechanism-enabled MAC protocol to access the shared channel. Considering random orientation of antennas, presence of blockages and MAC layer protocol, we derive the intensity of APs that actively introduce interference and contribute to the interference power level at a ``typical" receiving node. Subsequently, using the derived intensity, we obtain the Laplace transform of the interference power at a generic receiver (separated with an arbitrary distance from its associated AP) and then calculate the bit error rate (BER) expression and validate it using Monte-Carlo simulations of the network.


%
%

\section{System Model} \label{sec:SysMod}
As shown in Fig.~\ref{fig:Sys_model}, we consider a network of APs, distributed over $\mathbb{R}^2$, based on Poisson point process $\Phi_{\rm P}$ with intensity $\lambda$. For each typical user, its dedicated AP is referred to as the \emph{serving} AP and the rest of the APs act as \emph{interfering} APs. It has also been assumed that APs are equipped with directional antennas, steered towards their intended receiver, in uniformly random directions. The radiation pattern of directional antennas can be approximated by a triangular pattern, that is $g(\phi){=}\frac{1}{\varphi }, \,\,\, \phi {\in} [-\frac{\varphi}{2},\frac{\varphi}{2}]$, where $\varphi$ is the signal beamwidth. Maintaining a reasonable accuracy, this assumption provides mathematical simplicity and tractability~\cite{TH2017joonas,Yazdani2017Dir}.

We further assume that CSMA/CA protocol is employed as carrier sensing mechanism for APs to access a shared channel. Based on CSMA/CA MAC protocol, each AP senses the shared channel (medium) continuously, if the medium is clear (i.e., no other devices in a certain area around the transmitter is using the medium), it proceeds to transmit; otherwise, it postpones the transmission to another clear time slot. Roughly speaking, this MAC protocol determines which APs are allowed to transmit simultaneously given the fact that they are not within a specific contention area of each other.
It is worth mentioning that a typical AP may not sense the other APs if they are out of its sensing range (distance-wise), or in the sensing range but transmit in different directions or their signal is blocked by blockages.

We also model the spatial distribution of blockages as a Poisson point process with parameter $\rho$~\cite{Bai2014Blockage}. Due to the presence of blockages in the environment, not all the potential interfering APs actively contribute to the interference power level at a specific receiver. In fact, some of them are considered as blocked APs, as their signal is blocked by the obstacles in the environment.
It is worth noting that an interfering AP may be considered as blocked AP relative to the typical AP or its associated user or both. Having said that, an interfering AP may not be sensed by the typical serving AP due to the presence of blockages in the path between them. However, it may have a clear path (without blockage) toward the associated user or vice versa; it can even be blocked for both and our model captures all of these scenarios.
Therefore, considering antennas directivity, blockage effect and MAC protocol, the density of the interfering APs that actively contribute to the interference level differs from the primary Poisson with intensity $\lambda$. In the following section, we first calculate the density of the active interfering APs and then derive the interference distribution and evaluate the network performance based on the BER metric.
\begin{figure}[t]
\centering
\includegraphics[scale=0.4]{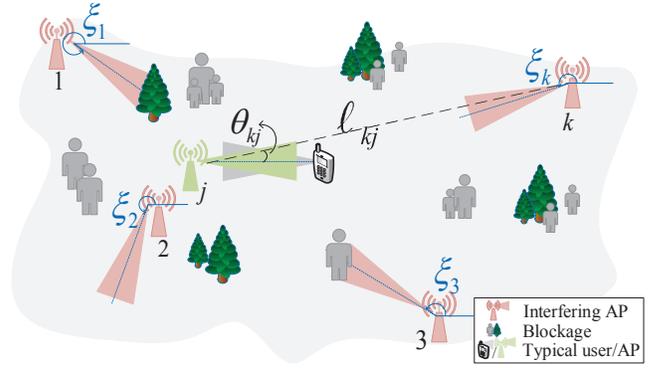}
\vspace{-0.5cm}
\caption{The impact of interferers on the victim receiver in the presence of obstacles (objects such as human bodies, trees and so on).}
\centering
\label{fig:Sys_model}
\end{figure}
\section{MAC Protocol Model, and Cross-layer Interference Analysis} \label{sec:ITbounds}
In order to calculate the density of the simultaneously interfering APs, we first describe the MAC protocol model considered in the cross-layer interference analysis.

\subsection{MAC protocol Model} \label{subsec:MAC_mdl}
We denote by $\Phi_{\rm M}$ the set of APs selected by MAC protocol to transmit simultaneously at a given time. In order to determine the set $\Phi_{\rm M}$ of the APs, each AP in the primary PPP process, $\Phi_{\rm P}$, is marked with a random number uniformly distributed in $[0,1]$. This mark abstracts the back-off timer~\footnote{If the medium is busy in the first transmission attempt, the node will wait for a random time before the second attempt. In fact, in each neighborhood of nodes, only node with the smallest back-off timer transmits.} in MAC protocol with collision avoidance (such as CSMA/CA) to prevent the collision. In addition, for $j^{\text{th}}$ AP, the set of neighbors is formally defined as APs in set $N_j$, such that
\begin{align} \label{eq:neighborSet}
\hspace{-0.26cm}{N_j}{=} \left\{ {\,k \in {\Phi _{P}}\left| {\,q_k\ell _{kj}^{ - \alpha }h_{kj}{\kern 1pt} g\left( {{\theta _{kj}}} \right){\kern 1pt} \breve{g}\left( {{\theta _{kj}} {+} \pi  {-} {\xi _k}} \right)z_k {>} \sigma } \right.\,} \right\}.
\end{align}
Here, $q_k$ is the transmit power of the $k^\text{th}$ AP. $\ell_{kj}$ represents the distance between the $j^{\text{th}}$ and $k^{\text{th}}$ APs. $\alpha$ and $h_{kj}$ denote the path loss exponent and squared fading gain of the generic Nakagami channel, with parameter $m$, from $k^\text{th}$ AP to the $j^\text{th}$ AP, respectively. In addition, $g\left({\theta _{kj}}\right)$ and $\breve{g}\left({{\theta _{kj}} + \pi  - {\xi _k}}\right)$ represent the antenna gain of $j^{\text{th}}$ and $k^{\text{th}}$ APs in the given directions in the arguments, respectively (see Fig.~\ref{fig:Sys_model}). ${\theta _{kj}}$ and ${\xi _k}$ denote the orientation of the $k^{\text{th}}$ AP with respect to $j^{\text{th}}$ AP and the boresight of $k^{\text{th}}$ AP antenna in its local orientation. Given the fact that $k^{\text{th}}$ interfering AP can be at any random location in the network with a random direction, we assume that ${\theta _{kj}}$ and ${\xi _k}$ are uniform random variables in $[-\pi,\pi]$. Moreover, $z_k$ is a random binary factor that determines whether the link between $j^{\text{th}}$ and $k^{\text{th}}$ APs is blocked by blockages; that is
\begin{align} \label{eq:z}
z_k = \left\{ \begin{array}{l}
1\,\,\,\,\,\,\,\,{\text{not-blocked with probability }} {p_k}\\
0\,\,\,\,\,\,\,\,\,{\text{blocked with probability }} 1-{p_k}.
\end{array} \right.
\end{align}
Based on~\eqref{eq:neighborSet}, the $k^{\text{th}}$ AP is considered as neighbor of $j^{\text{th}}$ AP, if the received power from $k^{\text{th}}$ AP is above carrier sensing threshold $\sigma$. Each transmitter competes only with its neighbors to access the shared medium, and the transmitter with the smallest back-off timer (smallest mark) proceeds to transmit, while others keep silent. Now, given the definition of the neighbor set in~\eqref{eq:neighborSet}, a generic AP belongs to the set $\Phi_M$ if and only if it has the lowest mark among its neighbors. This means that in each neighborhood~\footnote{For a generic AP, we use term \emph{neighborhood} to describe the area in which the APs in its neighbor-set are located.}, AP with the smallest mark transmits and the rest of the APs keep silent during the transmission. Therefore, the neighbors of $j^{\text{th}}$ AP reside within an arbitrary-shape area around it. We can bound the area by a disc with radius $r_{\text{cont}}$, where~\footnote{We drop the subscripts for notational simplicity, in the rest of the paper.} ${\left. {\Pr \left\{ {q\,r_{{\rm{cont}}}^{ - \alpha }\,h{\kern 1pt} \,g\left( \theta  \right){\kern 1pt} \breve{g} \left( {\theta  + \pi  - \xi } \right) \ge \sigma } \right\}} \right|_{\theta ,\xi }} \le \varepsilon $. Here, $\varepsilon$ is a small value. In fact, $r_{\text{cont}}$ is a sufficiently large distance beyond which the probability of an AP becoming a neighbor for an arbitrary AP $j$ is very negligible (smaller than $\varepsilon$). Therefore,
\begin{align}
{r_{{\rm{cont}}}} \le \frac{1}{ {4\pi ^2}}{\left( {\frac{{q{{\bar F}_h}^{-1}\left( \varepsilon  \right){\varphi ^{2\alpha  - 2}}}}{{\sigma}}} \right)^{\frac{1}{\alpha }}}.
\end{align}
Here, ${{\bar F}_h}^{-1}(.)$ is the inverse complementary cumulative distribution function (CCDF) of the squared fading gain of the channel. 
Determined by the MAC protocol, the average number of APs that concurrently transmit can be obtained by the following lemma.
\newtheorem{lemma}{Lemma}
\begin{lemma}\label{lem:lemma1}
Considering CSMA/CA protocol with sensing threshold $\sigma$, the average number of APs concurrently transmit, i.e., $\lambda_{\Phi_{\rm M}}$, is given by $\frac{{ \left( {1 - {{\rm{e}}^{ - \lambda A\eta }}} \right)}}{{\eta A}}$, where $A=\pi {r_{{\rm{cont}}}^2} $ is the contention area, and $\eta$ denotes the neighborhood success probability, defined in subsection~\ref{subsubsec:Neigh_SP}.
\end{lemma}
\begin{proof}
The density of concurrently transmitting APs, is the density of the primary $\Phi_{\rm P}$ thinned by the probability of retaining a generic AP in $\Phi_{\rm M}$. Moreover,
\begin{align} \notag
&\Pr \left\{ \text{retaining a generic AP in $\Phi_{\rm M}$} \right\}\\
&= \Pr \left\{ {\left. {{\text{retaining a generic AP in $ {\Phi _{\rm{M}}}$ }}} \right|{{\left| N \right|}} = i} \right\}\Pr \left\{ {{{\left| N \right|}} = i} \right\},
\end{align}
where $\Pr \left\{ {\left| N \right| = i} \right\}$ is probability of having $i$ neighbors; and
\begin{align} \notag
\Pr \left\{ {\left| N \right| = i} \right\} &= \sum\limits_n {\Pr \left\{ {\left| N \right| = i\,,\,\,\left| K \right| = n} \right\}}\\
&={\sum\limits_{n = 0}^\infty  {\frac{{{{\left( {\lambda A} \right)}^n}{{\rm{e}}^{ - \lambda A}}}}{{n!}} {\binom{n}{i}}{\eta ^i}{{\left( {1 - \eta } \right)}^{n - i}}} }.
\end{align}
Here, $ \left| K \right| $ denotes the number of AP in the contention area of a generic AP.
In fact, considering the PPP assumption of the distribution of APs, the probability of having $n$ APs in contention area $A$ around a generic AP is given by ${\frac{{{{\left( {\lambda A} \right)}^n}{{\rm{e}}^{ - \lambda A}}}}{{n!}}}$. However, given the pathloss, small scale fading, blockage effect and antenna directions, only $i$ out of $n$ APs are actually counted as neighbors, each of them with success probability $\eta$. In addition, with probability $\frac{1}{i+1}$ only one of them has the lowest mark and proceeds to transmit. Therefore,
\begin{align} \notag
&\Pr \left\{ \text{retaining a generic AP in $\Phi_{\rm M}$} \right\}\\ \notag
&\hspace{0.25cm}=\sum\limits_{i = 0}^n {\frac{1}{{i + 1}}\sum\limits_{n = 0}^\infty  {\frac{{{{\left( {\lambda A} \right)}^n}{{\rm{e}}^{ - \lambda A}}}}{{n!}}{\binom{n}{i}}{\eta^i}{{\left( {1 - \eta} \right)}^{n - i}}} }\\ \notag
&\hspace{0.25cm} = \sum\limits_{i = 0}^\infty  {\sum\limits_{n = i}^\infty  {\frac{{{{\rm{e}}^{ - \lambda A}}{\eta ^i}}}{{\left( {i + 1} \right)!}}{{\left( {1 - \eta } \right)}^{ - i}}\frac{{{{\left( {\lambda A} \right)}^n}}}{{n!}}\frac{{n!}}{{(n - i)!}}{{\left( {1 - \eta } \right)}^n}} } \\ \notag
&\hspace{0.25cm} = \frac{{1 - {{\rm{e}}^{ - \lambda A\eta}}}}{{ {\eta\lambda A} }}.
\end{align}
\end{proof}
%
\subsubsection{Neighborhood Success Probability} \label{subsubsec:Neigh_SP}
We define neighborhood success probability as
\begin{align} \notag
\eta  &= \int\limits_0^{{r_{{\rm{cont}}}}} {{\left. {{\rm{Pr}}\left\{ {q{\ell ^{ - \alpha }}h\,g\left( \theta  \right)\breve{g} \left( {\theta  + \pi  - \xi } \right)z > \sigma } \right\}} \right|_{\theta ,\xi }}{f_L}\left( \ell  \right){\rm{d}}\ell } \\
&\mathop  =  \,\frac{{{\varphi ^2}}}{{4{\pi ^2}}}\,\int\limits_0^{{r_{{\rm{cont}}}}} {{{\bar F}_h}\left( {\frac{{\sigma {\varphi ^2}{\ell ^\alpha }}}{{q\,}}} \right){f}\left( z  \right){f_L}\left( \ell  \right){\rm{d}}\ell } ,
\end{align}
where the distance distribution is
\begin{align} \label{eq:distance_dist}
{{f_L}\left( \ell  \right)}=\left\{ \begin{array}{l}
\frac{{2\ell}}{{{D^2}}}\,\,\,\,\,\,\,\,\,0 < \ell < D\\
0\,\,\,\,\,\,\,\,\,\,\,\,\,\text{elsewhere},
\end{array} \right.
\end{align}
in PPP network model. Moreover, ${f}\left( z  \right)$ is the distribution of the random variable $z$ in~\eqref{eq:z}.
\subsection{Blockage Model} \label{subsec:Blkg_mdl}
\begin{figure}[t]
\centering
\includegraphics[scale=0.5]{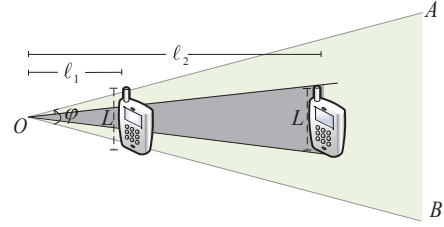}
\vspace{-0.5cm}
\caption{User$1$ (resp. User$2$) with distance $\ell_1$ (resp. $\ell_2$) to the transmitting AP. User$1$ (resp. User$2$) is blocked if at least one blockage intersect the triangle with area ${ {{\ell_1} ^2}\tan \left( \frac{\varphi}{2}  \right)}$ (resp. $\frac{\ell_2L}{2}$). Green area ($\bigtriangleup OAB$) represents the triangular beam pattern of the transmitting AP.}
\centering
\label{fig:Blockage}
\end{figure}
In order to calculate the blockage probability for $k^\text{th}$ AP, as mentioned previously, we approximate the radiation pattern of the antennas by a triangular-shaped pattern denoted by $C_k$, where its edges are determined by the signal beamwidth $\varphi$ (green triangular area in Fig.~\ref{fig:Blockage}). Considering the receiving terminal with average length $L$, smaller than blockage dimension~\footnote{Users' handheld units are smaller than blockages like human bodies, foliage, cars and so on.}, the signal from $k^\text{th}$ transmitter is blocked with at least one blockage in the line-of-sight path to the receiver.
%
%
Given the PPP assumption of the location of the blockages, for the $k^\text{th}$ AP, the probability of not being blocked (having line-of-sight) is given by (see Fig.~\ref{fig:Blockage})
\begin{align} \label{eq:p_k}
{p_k} {=} {{\rm{e}}^{ - \rho |{C_k}|}} {=} \left\{ \begin{array}{l}
{{\rm{e}}^{ - \rho {{\ell_k} ^2}\tan \left( \frac{\varphi}{2}  \right)}}\,\,\,\,\,\,\,\,\,0 \le \ell_k  < \frac{L}{{2\tan \left( \frac{\varphi}{2}  \right)}}\\
{{\rm{e}}^{ - \rho \frac{{\ell_k L}}{2}}}\,\,\,\,\,\,\,\,\,\,\,\,\,\,\,\,\,\,\,\,\,\,\,\,\,\,\ \ell_k  \ge \frac{L}{{2\tan \left( \frac{\varphi}{2}  \right)}},
\end{array} \right.
\end{align}
where $\ell_k$ is the distance from the $k^\text{th}$ AP to an arbitrary location from where its blockage probability is calculated. $\{\ell_1, \ell_2, ..., \ell_k,...\}$ is the sequence of distances of APs from an arbitrary location in the network, with distribution given by~\eqref{eq:distance_dist}.

\subsection{Interference Statistics} \label{subsec:Intf_mdl}

Given the intensity of the non-blocked APs that are concurrently allowed to transmit based on the MAC protocol in lemma~\ref{lem:lemma1}, we can derive the Laplace transform of the accumulated interference power at a typical receiving node in the network. Following from~\cite{Hdcore2011Haenggi}, the set of simultaneously transmitting APs can be safely approximated by a PPP. The intensity of the corresponding PPP is derived in lemma~\ref{lem:lemma1}.
\newtheorem{theorem}{Theorem}
\begin{theorem} \label{th:theorem}
The Laplace transform of the accumulated interference power, from the non-blocked APs that are concurrently allowed to transmit based on the MAC protocol, received at a typical receiving node, denoted by ${\mathcal{L}_{{I_{{\rm{agg}}}}}}\left( {\rm{s}} \right)$, is given by
\begin{align} \label{eq:MGF} \notag
&\hspace{-0.1cm}{\mathcal{L}_{{I_{{\rm{agg}}}}}}\left( {\rm{s}} \right) {=} \exp \Bigg\{  - \frac{{{\lambda _{{\Phi _M}}}{\varphi ^2}}}{{2\pi }}\bigg[ \frac{1}{{2\rho \tan \left( \varphi  \right)}} \\
&\hspace{1.1cm}+ \big( {\frac{4}{{{\rho ^2}{L^2}}} + \frac{1}{{2\rho \tan \left( \varphi  \right)}}} \big){e^{ - {\mkern 1mu} {\mkern 1mu} \frac{{\rho {L^2}}}{{4\tan \left( \varphi  \right)}}}} - {\kappa _m}\left( {\rm{s}} \right) \bigg] \Bigg\},
\end{align}
where
\begin{align} \label{eq:MGF_proof}\notag
&{\kappa _m}\left( {\rm{s}} \right) = \int\limits_0^{\frac{L}{{2\tan \left( \varphi  \right)}}}  {\ell {{\rm{e}}^{ - \rho {\ell ^2}\tan \left( \varphi  \right)}}{{\left( {1 + {\rm{s}}\frac{{q{\ell ^{ - \alpha }}}}{{m{\varphi ^2}}}} \right)}^{ - m}}}{\rm{d}}\ell  \\
&\hspace{2.5cm}- \int\limits_{\frac{L}{{2\tan \left( \varphi  \right)}}}^\infty  {{\ell {{\rm{e}}^{ - \rho \frac{{\ell L}}{2}}}{{\left( {1 + {\rm{s}}\frac{{q{\ell ^{ - \alpha }}}}{{m{\varphi ^2}}}} \right)}^{ - m}}}{\rm{d}}\ell }.
\end{align}
\end{theorem}
\begin{proof}
The Laplace transform of the accumulated interference power is defined as
\begin{align} \notag
{\mathcal{L}_{{I_{{\rm{agg}}}}}}\left( {\rm{s}} \right)\left| {_U} \right. &= {\mathbb{E}}\left[ {{{\rm{e}}^{{\rm{-s}}\sum\limits_{k = 1}^U {{q_k}\ell _{kj}^{ - \alpha }{h_{kj}}\,g\left( {{\theta _{kj}}} \right){\kern 1pt} \breve{g} \left( {{\theta _{kj}} + \pi  - {\xi _k}} \right){z_k}\left| {_{{\theta _{kj}},{\xi _k},{z_k}}} \right.} }}} \right]\\ \notag
& \mathop  = \limits^{(a)} \prod\limits_{k = 1}^U {{\mathbb{E}}\left[ {{{\rm{e}}^{{\rm{-s}}{q_k}\ell _{kj}^{ - \alpha }{h_{kj}}\,g\left( {{\theta _{kj}}} \right){\kern 1pt} \breve{g} \left( {{\theta _{kj}} + \pi  - {\xi _k}} \right){z_k}\left| {_{\theta_{kj} ,\xi_{k} ,z_k}} \right.}}} \right]} \\
& \mathop  = \limits^{(b)} {\left( {{\mathbb{E}}\big[ {{{\rm{e}}^{{\rm{-s}}{q}\ell^{ - \alpha }{h}\,g\left( {{\theta}} \right){\kern 1pt} \breve{g} \left( {{\theta} + \pi  - {\xi}} \right){z}\left| {_{\theta ,\xi ,z}} \right.}}} \big]} \right)^U}.
\end{align}
(a) follows from the independence of $\ell _{kj}$, $h_{kj}$, $\theta _{kj}$, $\xi _{k}$ and $z_k$. (b) utilizes the fact that the product of the sequence can be calculated with respect to the attributes of an arbitrary AP. In addition, for simplicity of notations, we drop the indices. In order to calculate~\eqref{eq:MGF}, we first consider a disk of radius $D$ and then take the limit as $D {\to} \infty$. Therefore, given the PPP assumption of the location of the APs,
\begin{align} \label{eq:Dinf}\notag
{\mathcal{L}_{{I_{{\rm{agg}}}}}}{(\rm{s})}&{=} \mathop {\lim }\limits_{D \to \infty } \sum\limits_{u = 0}^\infty  \frac{{{{\rm{e}}^{ - {\lambda _{{\Phi _M}}} \pi {D^2}}}{{\left( {{\lambda _{{\Phi _M}}} \pi {D^2}} \right)}^u}}}{{u!}}\\ \notag
&\hspace{1.5cm}\times{{\left( {{\mathbb{E}}\big[ {{{\rm{e}}^{{\rm{-s}}q{\ell ^{ - \alpha }}h\,{g} \left( \theta  \right)\,\breve{g}\left( {\theta  + \pi  - \xi } \right)z\left| {_{\theta ,\xi ,z}} \right.}}} \big]} \right)}^u}\\
&{=}\mathop {\lim }\limits_{D \to \infty }{{\rm{e}}^{ - {\lambda _{{\Phi _M}}} \pi {D^2}\left( {1 - {\mathbb{E}}\big[ {{{\rm{e}}^{{\rm{-s}}q{\ell ^{ - \alpha }}h\,{g} \left( \theta  \right)\,\breve{g}\left( {\theta  + \pi  - \xi } \right)z\left| {_{\theta ,\xi ,z}} \right.}}} \big]} \right)}}.
\end{align}
In order to calculate~\eqref{eq:Dinf}, we have
\begin{align} \notag
&\mathop {\lim }\limits_{D \to \infty } \,\,\,\, - {\lambda _{{\Phi _M}}} \pi {D^2}\left( {1 - {\mathbb{E}}\big[ {{{\rm{e}}^{{\rm{-s}}q{\ell ^{ - \alpha }}h\,{g} \left( \theta  \right)\,\breve{g}\left( {\theta  + \pi  - \xi } \right)z\left| {_{\theta ,\xi ,z}} \right.}}} \big]} \right)\\ \notag
&\hspace{0.5cm}= \mathop {\lim }\limits_{D \to \infty } \,\int\limits_{ - \frac{\varphi }{2}}^{\frac{\varphi }{2}} \int\limits_{\pi  + \theta  - \frac{\varphi }{2}}^{\pi  + \theta  + \frac{\varphi }{2}} \int\limits_0^\infty  \int\limits_0^D  - \lambda _{{\Phi _M}} \pi {D^2}( {1 - {{\rm{e}}^{{-\rm{s}}\frac{{q{\ell ^{ - \alpha }}h}}{{{\varphi ^2}}}\,}}} )\\
&\hspace{1cm}\times{{f}\left( z  \right)}{\frac{2\ell}{D^2}}\frac{1}{{{{4\pi} ^2}}}\frac{{{m^m}}}{{\Gamma \left( m \right)}}{h^{m - 1}}{{\rm{e}}^{ - mh}}{\rm{d}}h{\rm{d}}\theta {\rm{d}}\xi {\rm{d}}\ell.
\end{align}
After algebraic manipulation, we arrive at expression~\eqref{eq:MGF}.
\end{proof}
Following the approach in~\cite{Niknam2017Spatial-Spectral}, the Laplace transform of the aggregated interference power, derived in theorem~\ref{th:theorem}, is used to define the average BER expression,
\begin{align} \label{Average_BER} \notag
{\rm{BER}}_{\rm{ave}}&= \frac{1}{2} - \frac{{\sqrt c }}{\pi }\frac{{\Gamma (m + \frac{1}{2})}}{{\Gamma (m)}}\int_0^\infty  {\frac{{{{\kern 1pt} _1}{F_1}( m+\frac{1}{2};\frac{3}{2};-c{\rm{s}})}}{{\sqrt {\rm{s}} }}} \\
 &\hspace{1.25cm}\times {\mathcal{L}_{I_{\rm{agg}}}}\left( {-\frac{m}{{{q_0}\ell _0^{ - \alpha }}}{\rm{s}}} \right){{\rm{e}}^{-{\frac{{ m\sigma _n^2}}{{{q_0}\ell _0^{ - \alpha }}} } {\rm{s}}}}{\rm{ds}},
\end{align}
where $q_0$ and $\ell_0$ denote the transmit power and the distance from the typical user to its serving AP and $c$ is a constant that depends on the modulation type. In addition, ${\sigma_n}^2$ represents the power of the additive white Gaussian noise (AWGN). Moreover, ${{\kern 1pt} _1}{F_1}(; ;)$ and $\Gamma(.)$ are the confluent hypergeometric and gamma functions, respectively.
\section{Numerical Results} \label{sec:NumResults}
\begin{figure*}[!tbp]
 \centering
\begin{minipage}[b]{0.32\textwidth}
\centering
\includegraphics[width=6cm,height=4.5cm]{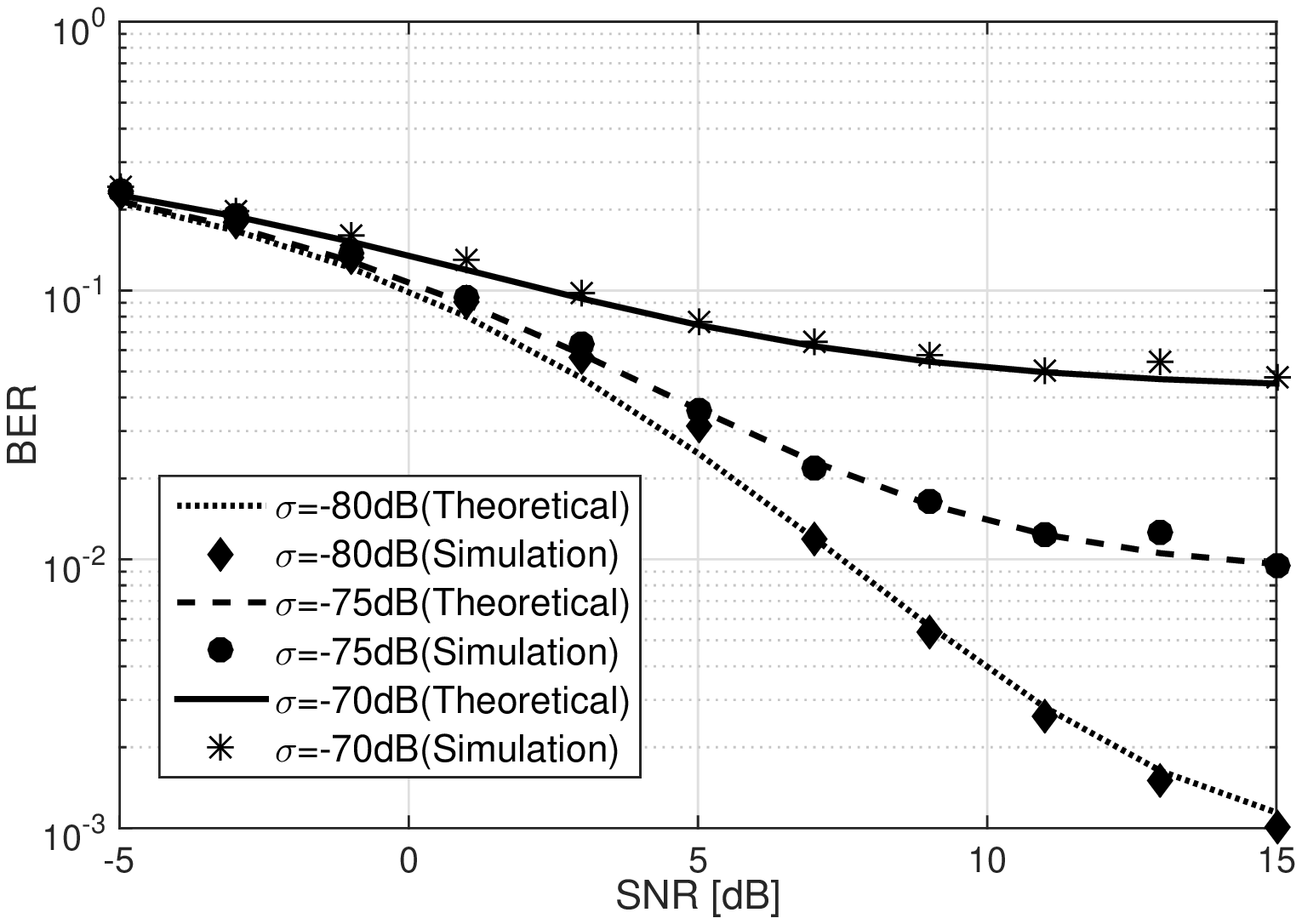}
\caption{BER Vs. SNR for different $\sigma$ values, $\lambda{=}10^{-1}$, $\rho{=}10^{-3}$.}
\label{fig:BER_SNR_sig}
 \vspace{-0.3cm}
   \end{minipage}
   \hfill
\begin{minipage}[b]{0.32\textwidth}
\centering
\includegraphics[width=6cm,height=4.5cm]{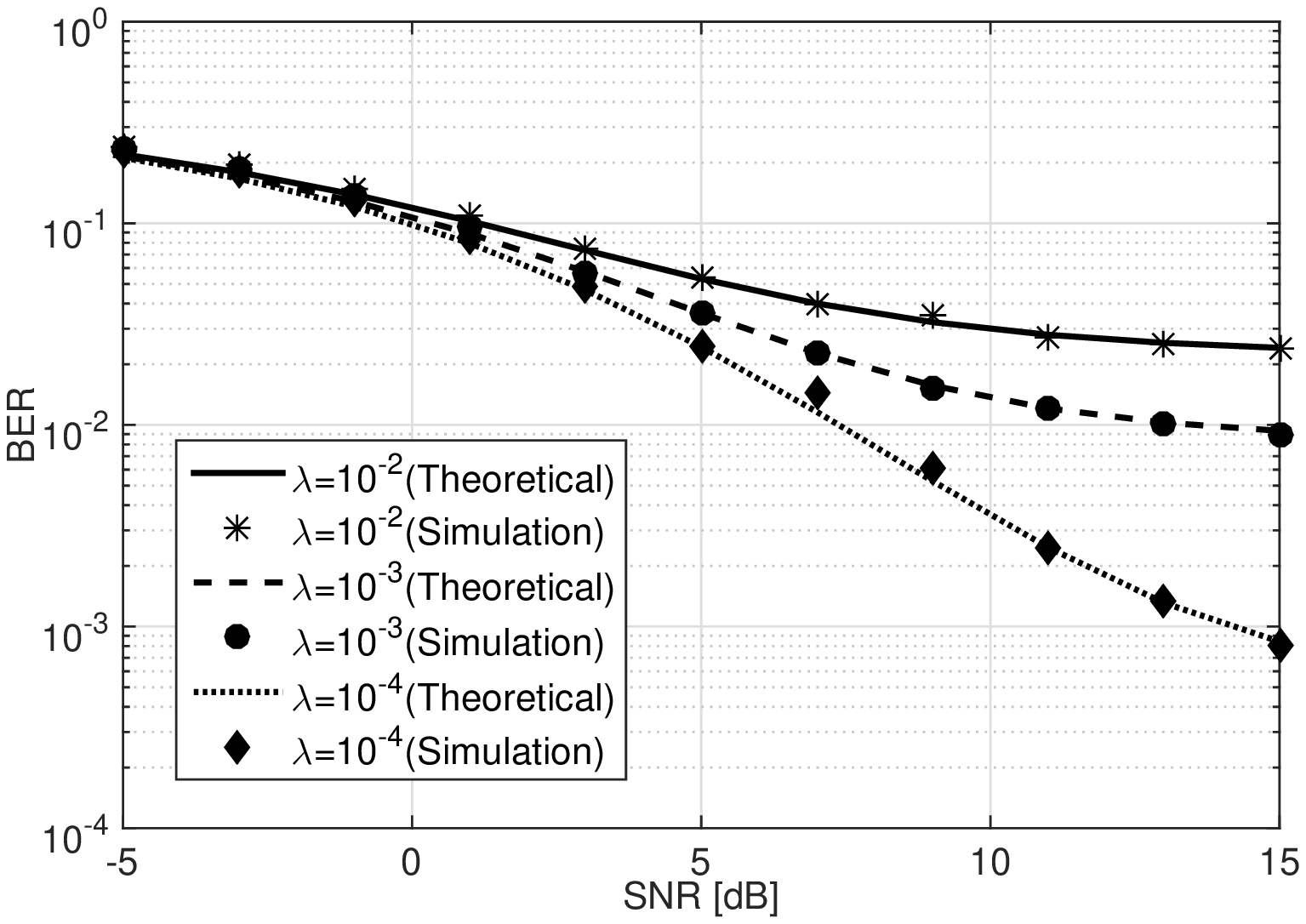}
\caption{BER Vs. SNR for different $\lambda$ values, $\rho{=}10^{-3}$, $\sigma{=}{-}60$ dB.}
\label{fig:BER_SNR_lambda}
 \vspace{-0.3cm}
   \end{minipage}
   \hfill
\begin{minipage}[b]{0.32\textwidth}
\centering
\includegraphics[width=6cm,height=4.5cm]{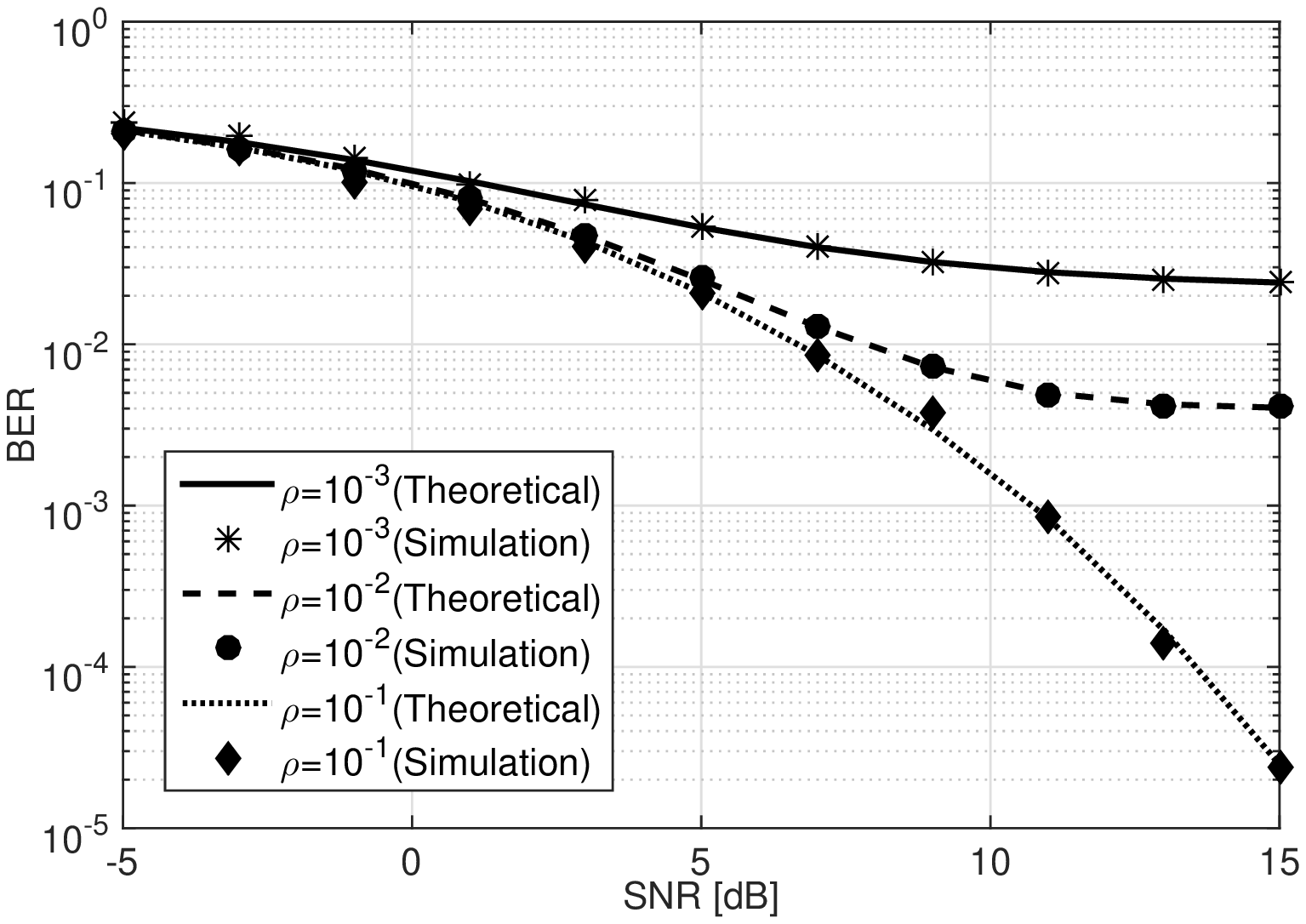}
\caption{BER Vs. SNR for different $\rho$ values, $\lambda{=}10^{-2}$, $\sigma{=}{-}60$ dB.}
\label{fig:BER_SNR_rho}
 \vspace{-0.3cm}
 \end{minipage}
\end{figure*}

This section provides numerical results to characterize the interference model. Monte-Carlo simulation validate the proposed model. We consider Nakagami-$m$ channel with shape factor $m=3$. Pathloss exponent $\alpha$ is set to $2.5$.  Here, the transmitted power of all interfering APs
are assumed to be the same and set to $30$ dBm. The beamwidth of the mmWave signals, i.e., $\varphi$, is assumed to be $15$ degrees. The distance between the serving AP and its intended receiver and the average receiver length $L$ are $5$m and $15$cm, respectively. Moreover, modulation parameter $c$ is $1$ (BPSK modulation).

Fig.~\ref{fig:BER_SNR_sig} represents the BER versus SNR curves for different carrier sensing threshold $\sigma$. As we can see, by increasing $\sigma$, the error performance increases. This is due to the fact that for larger $\sigma$ values the radius of the exclusion area around the transmitters $r_\text{cont}$ decreases. This means that when a typical transmitting node attempts to initiate the transmission and listen to the medium, it senses the signals from fewer number of other transmitting nodes. Therefore, higher number of nodes are allowed by MAC layer to transmit signals at the same time. These interfering nodes, on the other hand, can increase the interference power level at the receiver. This is an interesting result that shows, although mmWave signals can be easily blocked and attenuated, some level of carrier sensing is still needed in mmWave networks. In fact, in a dense network of APs, the interference power level can still be considerable. Therefore, in the average sense, sensing-based MAC protocols outperform the ALOHA-like ($\sigma \to \infty$) ones.

Fig.~\ref{fig:BER_SNR_lambda} shows the error performance for different primary APs density $\lambda$. As it is seen in Fig.~\ref{fig:BER_SNR_lambda}, increasing the primary density of the APs degrades the error performance. In fact, even with high directionality of mmWave signals and sensitivity to the obstacles, the density of the APs can not be carelessly increased to achieve higher data rate. The efficient number of APs, beyond which the error performance is higher than the desirable threshold, is captured in our model.

Fig.~\ref{fig:BER_SNR_rho} demonstrates the BER versus SNR for different blockage density $\rho$. With higher blockage density, larger number of interfering APs are blocked and hence the received interference power at the receiver decreases. In fact, considering beam directionality and hence sensitivity to blockages at mmWave frequencies, we might be able to design and operate a denser network of APs as captured by our model.
\section{Conclusion} \label{sec:Conclusion}
In this paper, we propose a cross-layer interference model for 5G mmWave networks. The derived model captures mmWave beam directionality, sensitivity to blockages and also MAC layer constraints. We propose a blockage model to account for the effect of obstacles in the environment. In addition, MAC protocol is considered in deriving the number of simultaneously transmitting APs. Subsequently, the Laplace transform of the power of the received interference at a typical user is derived and the error performance in terms of BER metric is evaluated and validated using simulations.


\bibliographystyle{IEEEtran}
\bibliography{IEEEabrv,GBbibfile}

\end{document}